\documentclass[conference]{IEEEtran}
\IEEEoverridecommandlockouts
\usepackage{ifpdf}
\usepackage{cite}
\usepackage[pdftex]{graphicx}
\usepackage{amsmath}
\usepackage{algorithmic}
\usepackage{array}
\usepackage{subfigure}
\usepackage{caption}
\usepackage[caption=false,font=footnotesize]{subfig}
\usepackage{fixltx2e}
\usepackage{stfloats}
\usepackage{url}
\usepackage{amsthm}
\usepackage{amsmath}
\usepackage{amsfonts}
\usepackage{amssymb}
\usepackage[T1]{fontenc} 
\usepackage{amsmath}
\usepackage[cmintegrals]{newtxmath}
\usepackage{bm} 
\usepackage[all]{xy}
\usepackage{enumitem}
\usepackage{float}
\usepackage{amsmath}
\usepackage[usenames, dvipsnames]{color}
\usepackage{ifpdf}
\usepackage{cite}
\usepackage[pdftex]{graphicx}
\usepackage{amsmath}
\usepackage{algorithmic}
\usepackage{array}
\usepackage{mathtools}
\DeclarePairedDelimiter{\abs}{\lvert}{\rvert}
\usepackage{amsmath}
\usepackage{amsfonts}
\usepackage{amssymb}
\usepackage{caption}
\usepackage[caption=false,font=footnotesize]{subfig}
\usepackage{fixltx2e}
\usepackage{stfloats}
\usepackage{url}
\usepackage{mathtools}
\usepackage{amsthm}
\usepackage{amsmath}
\usepackage{amsfonts}
\usepackage{amssymb}
\usepackage[T1]{fontenc} 
\usepackage{amsmath}
\usepackage[cmintegrals]{newtxmath}
\usepackage{bm} 
\usepackage[all]{xy}
\usepackage{enumitem}
\usepackage{float}
\usepackage{amsmath}
\usepackage[dvipsnames]{xcolor}
\usepackage{multirow}
\usepackage{subfig}
\usepackage{multicol}
\usepackage{graphicx}
\usepackage[caption=false]{subfig}
\usepackage{amssymb}
\usepackage{amsmath}
\usepackage{commath}
\usepackage{graphicx,bm}
\usepackage{verbatim}

\usepackage{diagbox}

\theoremstyle{definition}

\newtheorem{thm}{Theorem}

\newtheorem{define}{Definition}

\newcommand{\no}{\nonumber}

\begin{document}
	\title{Asymptotic Limits of Privacy in Bayesian Time Series Matching}	
\author{\IEEEauthorblockN{Nazanin Takbiri}
	\IEEEauthorblockA{Electrical and\\Computer Engineering\\
		UMass-Amherst\\
		ntakbiri@umass.edu}
	\and
	\IEEEauthorblockN{Dennis L. Goeckel}
	\IEEEauthorblockA{Electrical and\\Computer Engineering\\
		UMass-Amherst\\
		goeckel@ecs.umass.edu}
	\and
	\IEEEauthorblockN{Amir Houmansadr}
	\IEEEauthorblockA{Information and \\Computer Sciences\\
		UMass-Amherst\\
		amir@cs.umass.edu}
	\and
	\IEEEauthorblockN{Hossein Pishro-Nik}
	\IEEEauthorblockA{Electrical and\\Computer Engineering\\
		UMass-Amherst\\
		pishro@ecs.umass.edu\thanks{This work was supported by National Science Foundation under grants CCF--1421957 and CNS--1739462.}}
	}
\maketitle

\begin{abstract}	
Various modern and highly popular applications make use of user data traces in order to offer specific services, often for the purpose of improving the user's experience while using such applications. However, even when user data is privatized by employing privacy-preserving mechanisms (PPM), users' privacy may still be compromised by an external party who leverages statistical matching methods to match users' traces with their previous activities. In this paper, we obtain the theoretical bounds on user privacy for situations in which user traces are matchable to sequences of prior behavior, despite anonymization of data time series. We provide both achievability and converse results for the case where the data trace of each user consists of independent and identically distributed (i.i.d.) random samples drawn from a multinomial distribution, as well as the case that the users' data points are dependent over time and the data trace of each user is governed by a Markov chain model. 

\end{abstract}



\section{Introduction}
\label{intro}
\IEEEPARstart{T}{he} Internet of Things (IoT) is an important emerging technology and is growing at a rapid pace: by 2020, over 50 billion devices will be connected together as part of the IoT network \cite{IoT2}. Environmental monitoring, infrastructure management, energy management, medical and healthcare systems, building and home automation, and transport systems are some examples which indicate that IoT devices will affect nearly every aspect of our daily lives. However, this ubiquity of impact also raises grave privacy concerns. In particular, each IoT user in each application is generating a sequence of data that can be modeled as a random process; for example, in location-based services, each user is generating location traces. These sequences of data in IoT systems often contain sensitive information about users, such as their locations, health information, and hobbies. As a result, such huge amount of data generated by IoT devices can critically damage users' privacy, thereby providing a significant obstacle to the adaption of IoT applications. Thus, IoT privacy has drawn the attention of the research community~\cite{FTC2015, 3ukil2014iot, 4Hosseinzadeh2014} to investigate effective privacy-preserving mechanisms (PPMs).

PPMs are used to increase the assurance that private data is not accessible to third parties. Two promising classes of PPMs are \emph{identity perturbation} and \emph{data perturbation}~\cite{hoh2005protecting,freudiger2007mix, Naini2016, soltani2017towards,soltani2018invisible, shokri2012protecting, gruteser2003anonymous, bordenabe2014optimal}. The identity perturbation technique or anonymization is the process of hiding the true identity of the data owner~\cite{hoh2005protecting,freudiger2007mix, Naini2016, soltani2017towards,soltani2018invisible}. This technique removes personal identifiers or converts personally identifiable information into aggregated data. The data perturbation or obfuscation is the process of hiding the users' data by adding noise~\cite{shokri2012protecting, gruteser2003anonymous, bordenabe2014optimal}. However, perturbation techniques reduce utility to provide better privacy protection; thus, obtaining the optimum levels of anonymization and obfuscation is important.

In~\cite{Naini2016, matching}, a comprehensive analysis of the asymptotic (in the length of the time series) optimal matching of time series to source distributions is presented in a non-Bayesian setting, where the number of users is a fixed, finite value.
However, in~\cite{nazanin_ISIT2017,tifs2016, ciss2017, Nazanin_IT, nazanin_ISIT2018, ISIT18-longversion, Nazanin_WCNC2019}, a Bayesian setting was adopted in which the adversary has accurate prior distributions for user behavior through past observations or other sources, and the asymptotic limits of user privacy were obtained.

In addition, Li et al.~\cite{KeConferance2017} provide an optimal hypothesis test in the case where the adversary has training sequences from the group of users rather than the exact probability distribution. 

In this paper, we adopt the same setting as~\cite{KeConferance2017}; however, our work has significantly different flavor than that of~\cite{KeConferance2017}. First,~\cite{KeConferance2017} finds the optimal test in the non-asymptotic regime where there exist two users, while here, the asymptotic limits of user privacy for the case of a large number of users are obtained. Second,~\cite{KeConferance2017} obtains the necessary conditions for breaking privacy, while here, conditions for both perfect anonymity and no privacy are obtained. Third,~\cite{KeConferance2017} establishes the optimal test for the case with binary alphabets where each user's trace consists of independent and identically distributed (i.i.d.) samples drawn from a Bernoulli distribution, while here, we extend our results to the case where each user's trace is governed by i.i.d. random samples of a multinoulli distribution. We also extend our results for a more general Markov chain model.

The remainder of this paper is organized as follows. Section \ref{sec:framework} discusses the system model and the metrics used in the paper. Achievability and converse results for the two-state i.i.d.\ model are presented
in Section \ref{Two-state}, and their extensions to the $r$-state i.i.d.\ model are presented in Section \ref{r-state}. In addition, achievability and converse results for a more general Markov chain model are presented in Section \ref{sec:markov}. Section \ref{conclusion} provides some final conclusions and directions for future work.

\section{Framework}
\label{sec:framework}

We assume a system with $n$ users. Each user creates a length-$m$ sequence of data, which is denoted by $\textbf{X}_u$,
\[\textbf{X}_u = \left[
X_u(1), X_u(2), \cdots X_u(m)\right]^T, \ \ \ \textbf{X} =\left[\textbf{X}_{1}, \textbf{X}_{2}, \cdots, \textbf{X}_{n}\right],
\]
where $X_u(k)$ is the actual data point of user $u$ at time $k$. 
For each user, there also exists a length-$l$ sequence of its past behavior which is denoted as $\textbf{W}_u$,
\[\textbf{W}_u = \left[
W_u(1), W_u(2), \cdots W_u(l)\right]^T , \ \ \ \textbf{W} =\left[\textbf{W}_{1}, \textbf{W}_{2}, \cdots, \textbf{W}_{n}\right].
\]
where $W_u(k)$ is the observation of the prior behavior of user $u$ at time $k$.

The adversary has access to the observations of the prior users' behavior and wants to use this knowledge to break users' privacy despite the usage of some PPMs. As shown in Figure \ref{fig:learning}, an anonymization technique is employed in order to perturb the users' identity before the data is provided to the IoT application. In this figure, ${Y}_u(k)$ is the reported data point of user $u$ at time $k$ after applying anonymization; hence, the adversary observes
\[\textbf{Y}_u = \left[
Y_u(1), Y_u(2), \cdots Y_u(m)\right]^T , \ \ \ \textbf{Y} =\left[\textbf{Y}_{1}, \textbf{Y}_{2}, \cdots, \textbf{Y}_{n}\right].
\]
where $\textbf{Y}$ is the permuted version of $\textbf{X}$. 
\begin{figure}[h]
	\centering
	\includegraphics[width =1\linewidth]{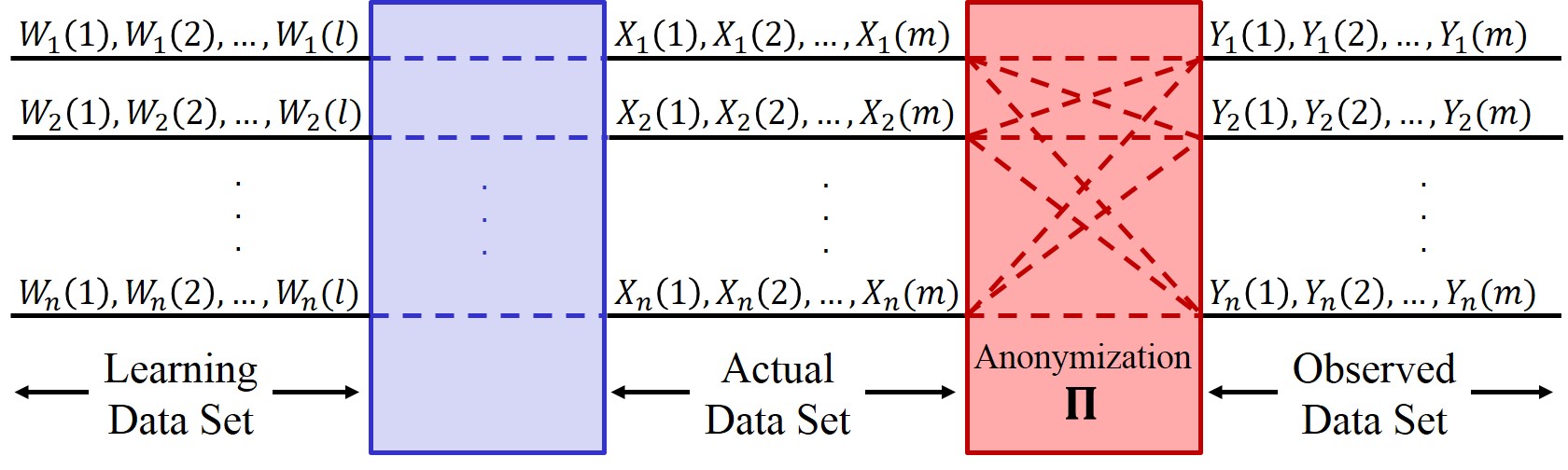}
	\caption{The goal of the adversary: match each sequence of $\textbf{W}_u$ of user $u\in \{ 1, 2, \cdots, n\}$ to an observed sequence $\textbf{Y}_u$ for $u\in \{1, 2, \cdots, n\}$.}
	\label{fig:learning}
\end{figure}

\subsection{Models and Metrics}
\textbf{{\textit{Data Points Model:}}}
We assume there exist $r$ possible values $\{0,1, \cdots, r-1\}$ for each data point. As shown in Figure \ref{fig:learning}, there exist two traces for each user: one that is termed "training data" and one that is termed "actual data," which needs to be protected from a malicious adversary. Remember that these two traces are generated from the same unknown probability distribution. In other words, for $k \in \{1,2, \cdots, m\}$ and $k' \in \{1, 2, \cdots, l\}$, both $X_u(k)$ and $W_u(k')$ are drawn from a user-specific probability distribution denoted as $\textbf{p}_u$.
 While all $\textbf{p}_u$'s are unknown to the adversary, each of them is drawn independently from a continuous density function $f_{\textbf{P}}(\textbf{x})$, where for all $x$ in the support of $f_{\textbf{P}}(\textbf{x})$, we assume
\begin{align}
\no 0 < \delta_1 <f_{\textbf{P}}(\textbf{x}) < \delta_2<\infty.\ \
\end{align}

\textbf{{\textit{Anonymization Mechanism:}}}
As shown in Figure \ref{fig:learning}, the mapping between users and data sequences is randomly permuted in order to achieve privacy. This random permutation is chosen uniformly at random among all $n!$ possible permutations on the set of $n$ users $\left(\mathbf{\Pi}:\{1, 2, \cdots, n\} \mapsto \{1, 2, \cdots, n\}\right)$; then, 
$\textbf{Y}_u=\textbf{X}_{\Pi^{-1}}$, $\textbf{Y}_{\Pi(u)}=\textbf{X}_{u}.$

\textbf{{\textit{Adversary Model:}}}
The adversary tries to match each sequence in the collection of training data traces $\left\{\textbf{W}_u, u=1, 2, \cdots, n\right\}$ with the sequence in the observation data traces $\left\{\textbf{Y}_u, u=1, 2, \cdots, n\right\}$ that is drawn from the same probability distribution, which we term \emph{statistical matching}. This is equivalent to finding the permutations of the user identities between two collections.
Note that the adversary knows the anonymization mechanism; however, he/she does not know the realization of the random permutation function.

Following~\cite{Nazanin_IT}, the definition of no privacy is as follows:
\begin{define}
	For an algorithm of the adversary that tries to estimate the actual data point of user $u$ at time $k$, define the error probability as
	\[P_e(u,k)= P\left(\widetilde{X_u(k)} \neq X_u(k)\right),\]
	where $X_u(k)$ is the actual data point of user $u$ at time $k$, and $\widetilde{X_u(k)}$ is the adversary's estimated data point of user $u$ at time $k$. Now, define ${\cal E}$ as the set of all possible estimators of the adversary. Then, user $u$ has \emph{no privacy} at time $k$, if and only if for large enough $n$,
	\[
	\no P^{*}_e(u,k)= \inf_{\cal E} {P\left(\widetilde{X_u(k)} \neq X_u(k)\right)} \to 0.
	\]
	Hence, a user has no privacy if there exists an algorithm for the adversary to estimate $X_u(k)$ with diminishing error probability as $n$ goes to infinity.
\end{define} 

In this paper, we also consider the situation in which there is perfect anonymity.
\begin{define}
	User $u$ has \emph{perfect anonymity} at time $k$ if and only if
	\begin{align}
	\no \lim\limits_{n\to \infty}
	\no H\left(\Pi(1)|\textbf{W}, \textbf{Y}\right) \to +\infty,
	\end{align}
where $H\left(\Pi(1)|\textbf{W}, \textbf{Y}\right)$ is the entropy of $\Pi(1)$ given $\textbf{W}$ and $\textbf{Y}$.
\end{define}

\section{Two-State i.i.d.\ Model}
\label{Two-state}

In this section, we assume each user's trace consists of samples from an i.i.d. random process and there are only two possible values for each user data point $X_u(k) \in \{0,1\}$. Thus, both training traces and real data traces are governed by an i.i.d. Bernoulli distribution with parameter $p_u$, where $p_u$ is probability that user $u$ taking value of a $1$, hence, 
$$W_u(k) \sim Bernoulli \left(p_u\right),$$
and
$$X_u(k) \sim Bernoulli \left(p_u\right), \ \ \ Y_u(k) \sim Bernoulli \left(p_{\Pi(u)}\right).$$
As discussed in Section \ref{sec:framework}, while $p_u$'s are unknown to the adversary, they are drawn independently from a known continuous density function ($f_P(x)$), where for all $x \in (0,1)$, we have 
\begin{align}
0 < \delta_1 <f_{{P}}(x) < \delta_2<\infty.\ \
\label{eq_f}
\end{align}

\subsection{Perfect Anonymity Analysis}
\label{two_perfect}
The following theorem states that if $m$ or $l$ are significantly smaller than $n^2$ in this two-state model, then all users have perfect anonymity.
\begin{thm}\label{two_perfect_thm}
	For the above two-state i.i.d.\ model, if $\textbf{Y}$ is the anonymized version of $\textbf{X}$, and $\textbf{W}$ is the past behavior of the users as defined above, and
	\begin{itemize}
		\item at least one of $m$ or $l$ is less than or equal to $cn^{2-\alpha}$ for any $c, \alpha>0$;
	\end{itemize}
	then, user $1$ has perfect anonymity at time $k$.
\end{thm}
\begin{proof}

First, consider the case $m \leq l$. Here, $\textbf{W}$ is considered as the training set and $\textbf{Y}$ is considered as the observed set; thus, 
given $\textbf{Y}$, $\textbf{W} \rightarrow \textbf{P} \rightarrow \Pi(1)$ forms a Markov chain. According to the data processing inequality,
\begin{align}
\no I\left(\Pi(1); \textbf{W}|\textbf{Y}\right) \leq I\left(\Pi(1);\textbf{P}|\textbf{Y}\right);
\end{align}
thus, 
\begin{align}
\no H(\Pi(1)|\textbf{Y})-H\left(\Pi(1)|\textbf{W}, \textbf{Y}\right) \leq H(\Pi(1)|\textbf{Y})-H\left(\Pi(1)|\textbf{P}, \textbf{Y}\right),
\end{align}	
and
\begin{align}
\no H\left(\Pi(1)|\textbf{W}, \textbf{Y}\right) \geq H\left(\Pi(1)|\textbf{P}, \textbf{Y}\right).
\end{align}		
In~\cite[Theorem 1]{tifs2016}, it is shown that if $m=n^{2-\alpha}$, $H\left(\Pi(1)|\textbf{P},\textbf{Y}\right) \to +\infty$, so, we can conclude
\begin{align}
\no H\left(\Pi(1)|\textbf{W}, \textbf{Y}\right) \to +\infty,
\end{align}	
as $n \to \infty$.

Now, consider the case $l \leq m$. By symmetry of the problem $\textbf{Y}$ can be considered as the training set and $\textbf{W}$ can be considered as the observed data. Thus, we can similarly prove the same results.
\end{proof}

\subsection{No Privacy Analysis}
\label{two_no}
The following theorem states that if both $m$ and $l$ are significantly larger than $n^2$ in this two-state model, then the adversary can find an algorithm to successfully estimate users' data points with arbitrarily small error probability, and as a result break users' privacy.

\begin{thm}\label{two_no_thm}
For the above two-state i.i.d.\ model, if $\textbf{Y}$ is the anonymized version of $\textbf{X}$, and $\textbf{W}$ is the past behavior of the users as defined above, and
\begin{itemize}
	\item $m=cn^{2+\alpha}$ for any $c, \alpha >0$;
	\item $l=c'n^{2+\alpha}$ for any $c', \alpha>0$;
\end{itemize}
then, user $1$ has no privacy at time $k$.
\end{thm}

\begin{proof}
For $u \in \{1, 2, \cdots n\}$, define
\[
\overline{{Y}_u}=\frac{Y_u(1)+Y_u(2)+ \cdots +Y_u(m)}{m},
\]
\[
\overline{{Y}_{\Pi(u)}}=\frac{X_u(1)+X_u(2)+ \cdots +X_u(m)}{m},
\]
and
\[
\overline{{W}_u}=\frac{W_u(1)+W_u(2)+ \cdots +W_u(l)}{l}.
\]
We claim that for $m =cn^{2 + \alpha}$, $l =c'n^{2 +\alpha}$ and large enough $n$:
\begin{enumerate}
	\item $\mathbb{P}\left(\ \abs{\overline{{Y}_{\Pi(1)}}-\overline{{W}_1}}\leq \Delta_n\right) \to 1,$
	\item $\mathbb{P}\left( \bigcup\limits_{u=2}^n \left\{\ \abs{\overline{{Y}_{\Pi(u)}}-\overline{{W}_1}}\leq \Delta_n \right\}\right) \to 0,$
\end{enumerate}
where $\Delta_n=n^{-(1+\frac{\alpha}{4})}$.
Thus, the adversary can match $\textbf{W}_1$ to $\textbf{Y}_{\Pi(1)}$.

\textit{\textbf{First Step:}} 
We want to show 
\begin{align}
\no \mathbb{P}\left(\ \abs{\overline{{X}_{u}}-\overline{{W}_u}}\leq \Delta_n\right) \to 1.
\end{align}
Note $\mathbb{E}[X_{u}(k)]=\mathbb{E}[W_u(k)]=p_u$, so as $n \to \infty$,
\begin{align}
\no \mathbb{P}\left(\ \abs{\overline{{X}_{u}}-\overline{{W}_u}}\geq \Delta_n\right)=& \mathbb{P}\left(\ \abs{\overline{{X}_{u}}-p_u-\overline{{W}_u}+p_u}\geq \Delta_n\right)\\
\nonumber &\hspace{-0.7 in}\leq \mathbb{P}\left(\ \abs{\overline{{X}_{u}}-p_u}+\abs{\overline{{W}_u}-p_u} \geq\Delta_n \right)\\
\nonumber &\hspace{-0.7 in}\leq \mathbb{P}\left(\left\{\ \abs{\overline{{X}_{u}}-p_u}\geq\frac{\Delta_n}{2}\right\} \bigcup \left\{\ \abs{\overline{{W}_u}-p_u}\geq\frac{\Delta_n}{2}\right\} \right)\\
\no &\hspace{-0.7 in}\leq \mathbb{P}\left(\ \abs{\overline{{X}_{u}}-p_u}\geq\frac{\Delta_n}{2}\right)+\mathbb{P}\left(\ \abs{\overline{{W}_u}-p_u}\geq\frac{\Delta_n}{2} \right)\\
\no &\hspace{-0.7 in}\leq 2e^{-\frac{m\Delta_n^2}{12p_u}}+2e^{-\frac{l\Delta_n^2}{12p_u}}\\
\no &\hspace{-0.7 in}= 2e^{-\frac{cn^{2+\alpha}\cdot n^{-2-\frac{\alpha}{2}}}{12p_u}}+2e^{-\frac{c'n^{2+\alpha}\cdot n^{-2-\frac{\alpha}{2}}}{12p_u}}\\
&\hspace{-0.7 in}= 2e^{-\frac{cn^{\frac{\alpha}{2}}}{12}}+2e^{-\frac{c'n^{\frac{\alpha}{2}}}{12}} \to 0,
\label{eq_imp}
\end{align}
where the first inequality follows from the fact that $\abs{a-b} \leq \abs{a}+\abs{b}$, and as a result, $\mathbb{P}\left(\ \abs{a-b} \geq \Delta_n\right) \leq \mathbb{P}\left(\ \abs{a}+\abs{b} \geq \Delta_n\right)$. The union bound yields the third inequality, and the fourth inequality follows from Chernoff bounds. Now, for u=1, we have
\[\mathbb{P}\left(\ \abs{\overline{{Y}_{\Pi(1)}}-\overline{{W}_1}}\leq \Delta_n\right) \to 1,\]
as $n \to \infty.$

\textit{\textbf{Second Step:}}
First, we show as $n \to \infty$,
\[\mathbb{P}\left(\bigcup\limits_{u=2}^n \left\{\ \abs{p_u-p_1} \leq 4\Delta_n\right\} \right)\rightarrow 0. \]
According to (\ref{eq_f}), for all $u \in \{2, 3, \cdots, n\}$, we have
\[\mathbb{P}\left( \ \abs{p_u-p_1} \leq 4\Delta_n\right)\leq 8 \Delta_n \delta_2,\]
and according to the union bound,
\begin{align}
\no \mathbb{P}\left( \bigcup\limits_{u=2}^n \left\{\ \abs{p_u-p_1} \leq 4 \Delta_n\right\} \right) &\leq \sum\limits_{u=2}^n \mathbb{P}\left(\ \abs{p_u-p_1} \leq 4 \Delta_n\right) \\
\nonumber &\leq 8n \Delta_n \delta_2\\
\nonumber &= 8 {n^{-\frac{\alpha}{4}}} \delta_2 \to 0,\ \
\end{align}
as $n \to \infty$. Thus, for $u \in \{2, 3, \cdots, n\}$, the distance between $p_u$ and $p_1$ is bigger than $4\Delta_n$ with high probability.

Next, we show as $n \to \infty$,
\[\mathbb{P}\left(\bigcup\limits_{u=2}^n \left\{\ \abs{\overline{{W}_u}-\overline{{W}_1}} \leq 2\Delta_n\right\} \right)\rightarrow 0. \]
Note for all $u \in \{1, 2, \cdots, n \}$, Chernoff bounds yields:
\begin{align}
\mathbb{P}\left(\ \abs{\overline{{W}_u}-p_u}\geq \Delta_n\right) \leq 2e^{-\frac{l\Delta_n^2}{3p_u}} \leq 2e^{-\frac{l\Delta_n^2}{3}}.
\label{Wu_Pu}
\end{align}
As a result, for $u=1$, we have
\begin{align}
\no \mathbb{P}\left(\ \abs{\overline{{W}_1}-p_1}\geq \Delta_n\right) \leq 2e^{-\frac{c'n^{\frac{\alpha}{2}}}{3}} \to 0,
\end{align}
as $n \to \infty$. In other words, with high probability, the distance between $\overline{W_1}$ and $p_1$ is less than $\Delta_n$.

Now, given the fact that the distance between all $p_u$'s and $p_1$ is bigger than $4\Delta_n$, and the fact that the distance between $\overline{W_1}$ and $p_1$ is less than $\Delta_n$, for all $u \in \{2, 3, \cdots, n\}$, we have
\begin{align}
\no \mathbb{P}\left(\ \abs{\overline{{W}_u}-\overline{{W}_1}}\leq 2\Delta_n\right) &\leq \mathbb{P}\left(\ \abs{\overline{{W}_u}-p_u}\geq \Delta_n \right)\\
\no &\leq 2e^{-\frac{l\Delta_n^2}{3}}.
\end{align}
Thus,
\begin{align}
\no \mathbb{P}\left(\bigcup\limits_{u=2}^n \left\{\ \abs{\overline{{W}_u}-\overline{{W}_1}} \leq 2\Delta_n\right\} \right) &\leq \sum\limits_{u=2}^n \mathbb{P}\left(\ \abs{\overline{{W}_u}-\overline{{W}_1}} \leq 2\Delta_n\right) \\
\nonumber &\leq 2ne^{-\frac{l\Delta_n^2}{3}}\\
\nonumber &=2ne^{-\frac{c'n^{\frac{\alpha}{2}}}{3}} \to 0,\ \
\end{align}
as $n \to \infty$.

Now, we claim that given the fact that the distances between each of the $\overline{{W}_u}$'s and $\overline{{W}_1}$ are bigger than $2\Delta_n$, we have $$\mathbb{P}\left( \bigcup\limits_{u=2}^n \left\{\ \abs{\overline{{X}_{u}}-\overline{{W}_1}}\leq \Delta_n\right\}\right) \to 0.$$
Note, using (\ref{eq_imp}), we have
\begin{align}
\no \mathbb{P}\left(\ \abs{\overline{{X}_{u}}-\overline{{W}_1}}\leq \Delta_n\right)&=\mathbb{P}\left(\ \abs{\overline{{X}_{u}}-\overline{{W}_{u}}}\geq \Delta_n\right)\\
\no &\leq 2e^{-\frac{cn^{\frac{\alpha}{2}}}{12}}+2e^{-\frac{c'n^{\frac{\alpha}{2}}}{12}}.\ \
\end{align}
Thus, by using union bound, we have
\begin{align}
\no \mathbb{P}\left( \bigcup\limits_{u=2}^n \left\{\ \abs{\overline{{X}_{u}}-\overline{{W}_{1}}}\leq \Delta_n\right\}\right) &\leq \sum\limits_{u=2}^n\mathbb{P}\left(\ \abs{\overline{{X}_{u}}-\overline{{W}_{u}}}\geq \Delta_n\right)\\
\no &\leq 2ne^{-\frac{cn^{\frac{\alpha}{2}}}{12}}+2ne^{-\frac{c'n^{\frac{\alpha}{2}}}{12}} \to 0,\ \
\end{align}
as $n \to \infty$.

After completing the first and second steps, we can conclude if $m=cn^{2+\alpha}$ and $l=c'n^{2+\alpha}$, users have no privacy as $n \to \infty$.
\end{proof}

\section{$r$-State i.i.d.\ Model}
\label{r-state}

In this section, we assume each user's trace consists of samples from an i.i.d.\ random process, and users' data points can have $r$ possibilities, where $X_u(k) \in \{0, 1, \cdots, r-1\}$. Thus, both training traces and real data traces are governed by an i.i.d. multinoulli distribution with parameter $\textbf{p}_u$, and
\[\textbf{p}_u =\left[
p_u(1), p_u(2), \cdots p_u(r-1)\right]^T , \ \ \ \textbf{p} =\left[\textbf{p}_{1}, \textbf{p}_{2}, \cdots, \textbf{p}_{n}\right].
\]
where $p_u(i)$ is the probability that a datum of user $u$ has value $i$.

As discussed in Section \ref{sec:framework}, while $\textbf{p}_u$'s are unknown to the adversary, they are drawn independently from a known continuous density function $f_\textbf{P}(\textbf{x})$, where for all $\textbf{x }\in \mathcal{R}_{\textbf{p}}$, 
\begin{align}
\no \mathcal{R}_{\textbf{p}} &= \bigg{\{} (x_1,x_2, \cdots, x_{r-1}) \in (0,1)^{r-1}:\\
\no &\ \ \ \ \ \ x_i > 0 , x_1+x_2+\cdots+x_{r-1} < 1,\ i=1, 2, \cdots, r-1\bigg{\}},
\end{align}
we have 
\begin{align}
0 < \delta_1 <f_{{\textbf{P}}}(\textbf{x}) < \delta_2<\infty.\ \
\end{align}

\subsection{Perfect Anonymity Analysis}
\label{r_perfect}

The following theorem states that if $m$ or $l$ are significantly smaller than $n^{\frac{2}{r-1}}$ in this $r$-state model, then all users have perfect anonymity.

\begin{thm}\label{r_perfect_thm}
	For the above $r$-state i.i.d.\ model, if $\textbf{Y}$ is the anonymized version of $\textbf{X}$, and $\textbf{W}$ is the past behavior of the users as defined above, and
	\begin{itemize}
\item at least one of $m$ or $l$ is less than or equal to $cn^{\frac{2}{r-1}-\alpha}$ for any $c, \alpha>0$;
	\end{itemize}
	then, user $1$ has perfect anonymity at time $k$.
\end{thm}

\begin{proof}
We can now repeat the similar reasoning as Theorem \ref{two_perfect_thm}; then, by using~\cite[Theorem 2]{tifs2016}, the proof is complete.
\end{proof}

\subsection{No Privacy Analysis}
\label{r_no}

The following theorem states that if both $m$ and $l$ are significantly larger than $n^{\frac{2}{r-1}}$ in this $r$-state model, then the adversary can find an algorithm to successfully estimate users' data points with arbitrarily small error probability, and as a result break users' privacy.

\begin{thm}\label{r_no_thm}
	For the above $r$-state i.i.d.\ model, if $\textbf{Y}$ is the anonymized version of $\textbf{X}$, and $\textbf{W}$ is the past behavior of the users as defined above, and
	\begin{itemize}
		\item $m=cn^{\frac{2}{r-1}+\alpha}$ for any $c, \alpha >0$;
		\item $l=c'n^{\frac{2}{r-1}+\alpha}$ for any $c', \alpha>0$;
	\end{itemize}
	then, user $1$ has no privacy at time $k$.
\end{thm}

\begin{proof}
The proof of Theorem \ref{r_no_thm} is similar to the proof of Theorem \ref{two_no_thm}, so we just provide the general idea. We similarly define the empirical probability that the user with pseudonym $u$ has data sample $i$ as follows:

\[
\overline{Y_u(i)}=\frac{\abs {\left\{k \in \{1, 2, \cdots, m\}:Y_u(k)=i\right\}}}{m},
\]
and
\[
\overline{Y_{\Pi(u)}(i)}=\frac{\abs {\left\{k \in \{1, 2, \cdots, m\}:X_u(k)=i\right\}}}{m}.
\]
We also have 
\[
\overline{W_u(i)}=\frac{\abs {\left\{k \in \{1, 2, \cdots, l\}:W_u(k)=i\right\}}}{l}.
\]
The difference from the proof of Theorem \ref{two_no_thm} is that, for each $u \in \{1,2,\cdots, n \}$, $\overline{\textbf{Y}_{u}}$ and $\overline{\textbf{W}_{u}}$ are vectors of length $r-1$. In other words,
\[ \overline{\textbf{Y}_u} =\left[\overline{Y_u(1)}, \overline{Y_u(2)}, \cdots, \overline{Y_u(r-1)}\right]^T, \ \ \ \ u\in \{1, 2, \cdots, n\},\]
\[ \overline{\textbf{W}_u} =\left[\overline{W_u(1)}, \overline{W_u(2)}, \cdots, \overline{W_u(r-1)}\right]^T, \ \ \ \ u\in \{1, 2, \cdots, n\},\]
and we claim for $m = cn^{\frac{2}{r-1} + \alpha}$, $l = c'n^{\frac{2}{r-1}+\alpha}$, and large enough $n$,
\begin{enumerate}
	\item $\mathbb{P}\left(\ \abs{\overline{{\textbf{Y}}_{\Pi(1)}}-\overline{{\textbf{W}}_1}}\leq \Delta'_n\right) \to 1,$
	\item $\mathbb{P}\left( \bigcup\limits_{u=2}^n \left\{\ \abs{\overline{{\textbf{Y}}_{\Pi(u)}}-\overline{{\textbf{W}}_1}}\leq \Delta'_n \right\}\right) \to 0,$
\end{enumerate}
where $\Delta'_n=n^{-\left(\frac{1}{r-1}+\frac{\alpha}{4}\right)}$.
\end{proof}

\section{$r$-State Markov Chain Model}
\label{sec:markov}
In Section \ref{Two-state} and \ref{r-state}, the data trace of each user is governed by an i.i.d.\ random process, while here the data trace of each user is governed by an
irreducible and aperiodic $r$-state Markov chain where $E$ is the set of edges. Let us define the transition probability from state $i$ to state $j$ as:
$$p_u(i,j)=\mathbb{P}\left(X_u(k+1)=j|X_u(k)=i\right);$$
thus, $(i, j) \in E$ if and only if $p_u(i,j)>0$.

Here, we assume the same Markov chain structure for all of the users, but different users have different transition matrices. Note that a subset of the transition probabilities with size $|E|-r$ is sufficient for recovering the whole transition matrix. Let this subset be called $\textbf{p}_u$, so
\[\textbf{p}_u = \left[
p_u(1), p_u(2), \cdots p_u(|E|-r)\right]^T , \ \ \ \textbf{p} =\left[\textbf{p}_{1}, \textbf{p}_{2}, \cdots, \textbf{p}_{n}\right].
\]
where $p_u(i)$ is the probability that a datum of user $u$ has value $i$.
As discussed in Section \ref{sec:framework}, while $\textbf{p}_u$'s are unknown to the adversary, they are drawn independently from a known continuous density function $f_\textbf{P}(\textbf{x})$, where for all $\textbf{x }\in \mathcal{R}_{\textbf{p}}$, 
\begin{align}
\no \mathcal{R}_{\textbf{p}} &= \bigg{\{} (x_1,x_2, \cdots, x_{|E|-r}) \in (0,1)^{|E|-r}:\\
\no &\ \ \ x_i > 0 , x_1+x_2+\cdots+x_{|E|-r} < 1,\ i=1, 2, \cdots, |E|-r\bigg{\}},
\end{align}
we have 
\begin{align}
0 < \delta_1 <f_{{\textbf{P}}}(\textbf{x}) < \delta_2<\infty.\ \
\end{align}

\subsection{Perfect Anonymity Analysis}
\label{markov_perfect}

The following theorem states that if $m$ or $l$ are significantly smaller than $n^{\frac{2}{|E|-r}}$ in this $r$-state Markov chain model, then all users have perfect anonymity.

\begin{thm}\label{markov_perfect_thm}
	For the above $r$-state Markov chain model, if $\textbf{Y}$ is the anonymized version of $\textbf{X}$, and $\textbf{W}$ is the past behavior of the users as defined above, and
	\begin{itemize}
		\item at least one of $m$ or $l$ is less than or equal to $cn^{\frac{2}{|E|-r}-\alpha}$ for any $c, \alpha>0$;
	\end{itemize}
	then, user $1$ has perfect anonymity at time $k$.
\end{thm}

\begin{proof}
	We can now repeat the similar reasoning as Theorem \ref{two_perfect_thm}; then, by using~\cite[Theorem 3]{tifs2016}, the proof is complete.
\end{proof}

\subsection{No Privacy Analysis}
\label{markov_no}

The following theorem states that if both $m$ and $l$ are significantly larger than $n^{\frac{2}{|E|-r}}$, then the adversary can find an algorithm to successfully estimate users' data points with arbitrarily small error probability, and as a result, break users' privacy.

\begin{thm}\label{markov_no_thm}
	For the above $r$-state Markov chain model, if $\textbf{Y}$ is the anonymized version of $\textbf{X}$, and $\textbf{W}$ is the past behavior of the users as defined above, and
	\begin{itemize}
		\item $m=cn^{\frac{2}{|E|-r}+\alpha}$ for any $c, \alpha >0$;
		\item $l=c'n^{\frac{2}{|E|-r}+\alpha}$ for any $c', \alpha>0$;
	\end{itemize}
	then, user $1$ has no privacy at time $k$.
\end{thm}

\begin{proof}
	The proof of Theorem \ref{markov_no_thm} is similar to the proof of Theorem \ref{two_no_thm}, so we just provide the general idea. For each $u \in \{1,2,\cdots, n \}$, we similarly define $\overline{\textbf{Y}_{u}}$ and $\overline{\textbf{W}_{u}}$ as vectors of length $|E|-r$:
	\[ \overline{\textbf{Y}_u} =\left[\overline{Y_u(1)}, \overline{Y_u(2)}, \cdots, \overline{Y_u(|E|-r)}\right]^T, \ \ \ \ u\in \{1, 2, \cdots, n\}.\]
	\[ \overline{\textbf{W}_u} =\left[\overline{W_u(1)}, \overline{W_u(2)}, \cdots, \overline{W_u(|E|-r)}\right]^T, \ \ \ \ u\in \{1, 2, \cdots, n\}.\]
	We claim that for $m = cn^{\frac{2}{|E|-r} + \alpha}$, $l = c'n^{\frac{2}{|E|-r}+\alpha}$, and large enough $n$,
	\begin{enumerate}
		\item $\mathbb{P}\left(\ \abs{\overline{{\textbf{Y}}_{\Pi(1)}}-\overline{{\textbf{W}}_1}}\leq \Delta''_n\right) \to 1,$
		\item $\mathbb{P}\left( \bigcup\limits_{u=2}^n \left\{\ \abs{\overline{{\textbf{Y}}_{\Pi(u)}}-\overline{{\textbf{W}}_1}}\leq \Delta''_n \right\}\right) \to 0,$
	\end{enumerate}
	where $\Delta''_n=n^{-\left(\frac{1}{|E|-r}+\frac{\alpha}{4}\right)}$.
\end{proof}


\section{Conclusion}
\label{conclusion}

In this paper, we have derived the theoretical bounds on user privacy in situations in which user traces are matchable to prior user behavior despite anonymization protection. In particular, the adversary employs statistical matching of the user traces to previous behavior of users within a network to compromise their privacy.

As shown in Figure \ref{fig:result}, which displays the characterized privacy limits for the i.i.d.\ case, we demonstrated that the parameter plane, with coordinates length of learning set ($l$) and length of observed set ($m$), can be divided into two regions: in the first region, all users have perfect anonymity and in the second region no user has any privacy whatsoever. Specifically, we showed that if either $l$ or $m$ is significantly smaller than $n^{\frac{2}{r-1}}$, users have perfect anonymity and the adversary cannot identify the permutation function $\left(\mathbf{\Pi}\right)$, and, if both of them are significantly larger than $n^{\frac{2}{r-1}}$, users have no privacy. It is worth noting that in the case the adversary has the accurate prior information, which is discussed in \cite{tifs2016,ciss2017} and is shown in Figure \ref{fig:result2}, users have no privacy as long as number of adversary observations per user $m$ is larger than $n^{\frac{2}{r-1}}$.

For the case where the users' data points are governed by an
irreducible and aperiodic $r$-state Markov chain with $|E|$ edges, we demonstrated similar results: if either $l$ or $m$ is significantly smaller than $n^{\frac{2}{|E|-r}}$, users have perfect anonymity, and, if both of them are significantly larger than $n^{\frac{2}{|E|-r}}$, users have no privacy.

\begin{figure}[h]
	\centering
	\includegraphics[width =0.75\linewidth]{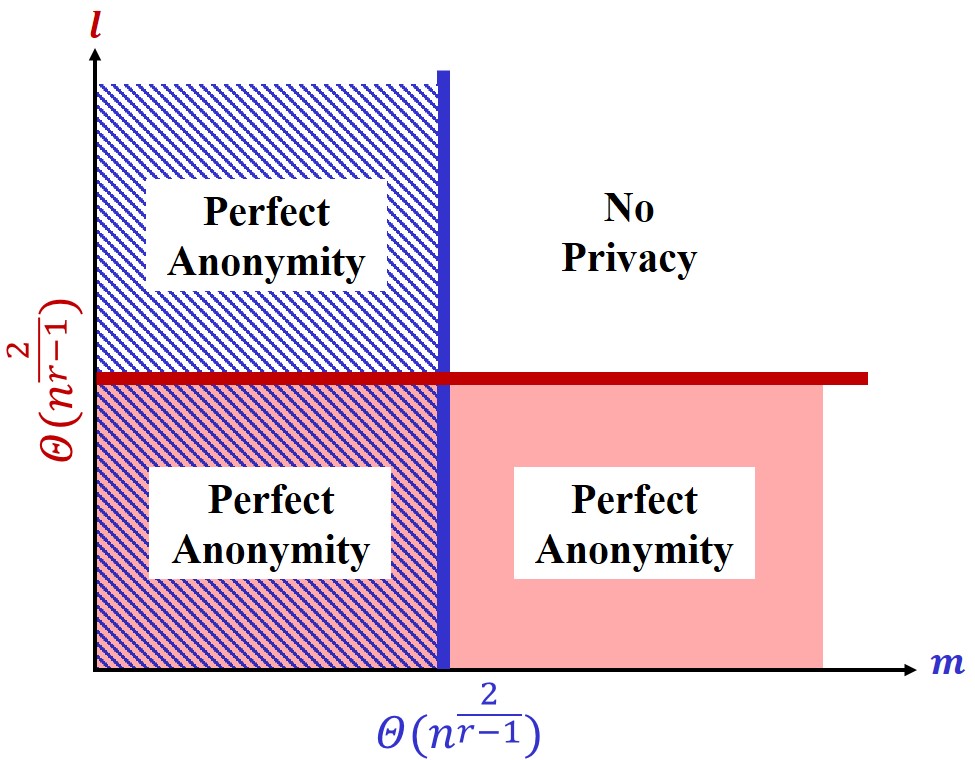}
	\caption{Limits of privacy in the entire $m-l$ plane in the case the adversary does not have the accurate prior distribution. Here, both training data traces and observed data traces are governed by an i.i.d.\ multinoulli distribution. $l$ is the length of the learning set, $m$ is the length of the observed data, and $r$ is the number of possible values for each user's data point.}
	\label{fig:result}
\end{figure}
\begin{figure}[h]
	\centering
	\includegraphics[width =0.75\linewidth]{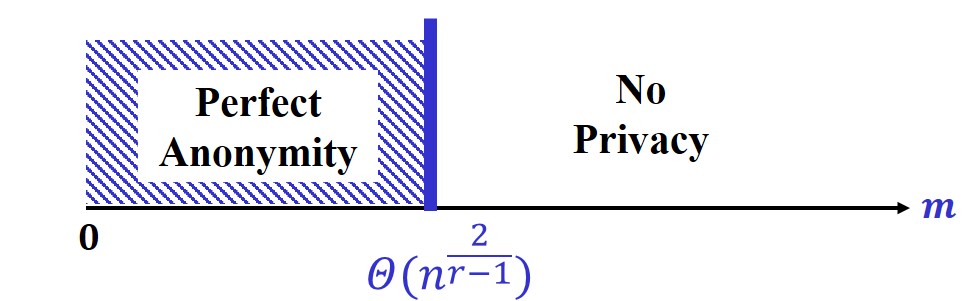}
	\caption{Limits of privacy in the case the adversary has an accurate prior distribution. Here, the observed data traces are governed by an i.i.d.\ multinoulli distribution, $m$ is the length of the observed data and $r$ is the number of possible values for each user's data point.}
	\label{fig:result2}
\end{figure}

\appendices

\bibliographystyle{IEEEtran}
\bibliography{REF,dennis_privacy}

\begin{thebibliography}{10}
\providecommand{\url}[1]{#1}
\csname url@samestyle\endcsname
\providecommand{\newblock}{\relax}
\providecommand{\bibinfo}[2]{#2}
\providecommand{\BIBentrySTDinterwordspacing}{\spaceskip=0pt\relax}
\providecommand{\BIBentryALTinterwordstretchfactor}{4}
\providecommand{\BIBentryALTinterwordspacing}{\spaceskip=\fontdimen2\font plus
\BIBentryALTinterwordstretchfactor\fontdimen3\font minus
  \fontdimen4\font\relax}
\providecommand{\BIBforeignlanguage}[2]{{%
\expandafter\ifx\csname l@#1\endcsname\relax
\typeout{** WARNING: IEEEtran.bst: No hyphenation pattern has been}%
\typeout{** loaded for the language `#1'. Using the pattern for}%
\typeout{** the default language instead.}%
\else
\language=\csname l@#1\endcsname
\fi
#2}}
\providecommand{\BIBdecl}{\relax}
\BIBdecl

\bibitem{IoT2}
\BIBentryALTinterwordspacing
J.~Bausch. (2016) The internet of things forecast of 50 billion connected
  devices by 2020 is grossly over-estimated and entirely misleading. [Online].
  Available:
  \url{https://www.electronicproducts.com/Internet_of_Things/Research/The_Internet_of_Things_forecast_of_50_billion_connected_devices_by_2020_is_grossly_over_estimated_and_entirely_misleading.aspx}
\BIBentrySTDinterwordspacing

\bibitem{FTC2015}
{Federal Trade Commission Staff}, ``Internet of things: Privacy and security in
  a connected world,'' 2015.

\bibitem{3ukil2014iot}
A.~Ukil, S.~Bandyopadhyay, and A.~Pal, ``{I}o{T}-privacy: To be private or not
  to be private,'' in \emph{IEEE Conference on Computer Communications
  Workshops (INFOCOM WKSHPS)}.\hskip 1em plus 0.5em minus 0.4em\relax Toronto,
  ON, Canada: IEEE, 2014, pp. 123--124.

\bibitem{4Hosseinzadeh2014}
S.~Hosseinzadeh, S.~Rauti, S.~Hyrynsalmi, and V.~Lepp{\"a}nen, ``Security in
  the internet of things through obfuscation and diversification,'' in
  \emph{IEEE Conference on Computing, Communication and Security
  (ICCCS)}.\hskip 1em plus 0.5em minus 0.4em\relax Pamplemousses, Mauritius:
  IEEE, 2015, pp. 1--5.

\bibitem{hoh2005protecting}
B.~Hoh and M.~Gruteser, ``Protecting location privacy through path confusion,''
  in \emph{First International Conference on Security and Privacy for Emerging
  Areas in Communications Networks (SecureComm)}.\hskip 1em plus 0.5em minus
  0.4em\relax Pamplemousses, Mauritius: IEEE, 2005, pp. 194--205.

\bibitem{freudiger2007mix}
J.~Freudiger, M.~Raya, M.~F{\'e}legyh{\'a}zi, P.~Papadimitratos, and J.~P.
  Hubaux, ``Mix-zones for location privacy in vehicular networks,'' Vancouver,
  2007.

\bibitem{Naini2016}
F.~M. Naini, J.~Unnikrishnan, P.~Thiran, and M.~Vetterli, ``Where you are is
  who you are: User identification by matching statistics,'' \emph{IEEE
  Transactions on Information Forensics and Security}, vol.~11, no.~2, pp.
  358--372, 2016.

\bibitem{soltani2017towards}
R.~Soltani, D.~Goeckel, D.~Towsley, and A.~Houmansadr, ``Towards provably
  invisible network flow fingerprints,'' in \emph{51th Asilomar Conference on
  Signals, Systems and Computers}, Pacific Grove, CA, USA, 2017.

\bibitem{soltani2018invisible}
R.~Soltani, D.~Goeckel, D.~F. Towsley, and A.~Houmansadr, ``Fundamental limits
  of invisible flow fingerprinting,'' \emph{CoRR}, vol. abs/1809.08514, 2018.

\bibitem{shokri2012protecting}
R.~Shokri, G.~Theodorakopoulos, C.~Troncoso, J.-P. Hubaux, and J.-Y. Le~Boudec,
  ``Protecting location privacy: optimal strategy against localization
  attacks,'' in \emph{Proceedings of the 2012 ACM conference on Computer and
  communications security}.\hskip 1em plus 0.5em minus 0.4em\relax ACM, 2012,
  pp. 617--627.

\bibitem{gruteser2003anonymous}
M.~Gruteser and D.~Grunwald, ``Anonymous usage of location-based services
  through spatial and temporal cloaking,'' in \emph{Proceedings of the 1st
  international conference on Mobile systems, applications and services}.\hskip
  1em plus 0.5em minus 0.4em\relax San Francisco, California, USA: ACM, 2003,
  pp. 31--42.

\bibitem{bordenabe2014optimal}
N.~E. Bordenabe, K.~Chatzikokolakis, and C.~Palamidessi, ``Optimal
  geo-indistinguishable mechanisms for location privacy,'' in \emph{Proceedings
  of the 2014 ACM SIGSAC Conference on Computer and Communications
  Security}.\hskip 1em plus 0.5em minus 0.4em\relax Scottsdale, Arizona, USA:
  ACM, 2014, pp. 251--262.

\bibitem{matching}
J.~Unnikrishnan, ``Asymptotically optimal matching of multiple sequences to
  source distributions and training sequences,'' \emph{IEEE Transactions on
  Information Theory}, vol.~61, no.~1, pp. 452--468, 2014.

\bibitem{nazanin_ISIT2017}
N.~Takbiri, A.~Houmansadr, D.~L. Goeckel, and H.~Pishro-Nik, ``Limits of
  locatin privacy under anonymization and obfuscation,'' in \emph{International
  Symposium on Information Theory (ISIT)}.\hskip 1em plus 0.5em minus
  0.4em\relax Aachen, Germany: IEEE, 2017, pp. 764--768.

\bibitem{tifs2016}
Z.~Montazeri, A.~Houmansadr, and H.~Pishro-Nik, ``{Achieving Perfect Location
  Privacy in Wireless Devices Using Anonymization},'' \emph{IEEE Transaction on
  Information Forensics and Security}, vol.~12, no.~11, pp. 2683--2698, 2017.

\bibitem{ciss2017}
N.~Takbiri, A.~Houmansadr, D.~Goeckel, and H.~Pishro-Nik, ``Fundamental limits
  of location privacy using anonymization,'' in \emph{51st Annual Conference on
  Information Science and Systems (CISS)}.\hskip 1em plus 0.5em minus
  0.4em\relax Baltimore, MD, USA: IEEE, 2017.

\bibitem{Nazanin_IT}
N.~Takbiri, A.~Houmansadr, D.~L. Goeckel, and H.~Pishro{-}Nik, ``Matching
  anonymized and obfuscated time series to users' profiles,'' \emph{IEEE
  Transactions on Information Theory}, vol.~65, no.~2, pp. 724--741, 2019.

\bibitem{nazanin_ISIT2018}
N.~Takbiri, A.~Houmansadr, D.~L. Goeckel, and H.~Pishro-Nik, ``Privacy against
  statistical matching: Inter-user correlation,'' in \emph{International
  Symposium on Information Theory (ISIT)}.\hskip 1em plus 0.5em minus
  0.4em\relax Vail, Colorado, USA: IEEE, 2018, pp. 1036--1040.

\bibitem{ISIT18-longversion}
N.~Takbiri, A.~Houmansadr, D.~L. Goeckel, and H.~Pishro{-}Nik, ``Privacy of
  dependent users against statistical matching,'' \emph{submitted to IEEE
  Transactions on Information Theory, Available at
  https://arxiv.org/abs/1710.00197}.

\bibitem{Nazanin_WCNC2019}
N.~Takbiri, R.~Soltani, D.~Goeckel, A.~Houmansadr, and H.~Pishro-Nik,
  ``Asymptotic loss in privacy due to dependency in gaussian traces,'' in
  \emph{IEEE Wireless Communications and Networking Conference (WCNC)}.\hskip
  1em plus 0.5em minus 0.4em\relax Marrakech, Morocco: IEEE, 2019.

\bibitem{KeConferance2017}
K.~Li, H.~Pishro-Nik, and D.~Goeckel, ``Bayesian time series matching and
  privacy,'' in \emph{51th Asilomar Conference on Signals, Systems and
  Computers}, Pacific Grove, CA, USA, 2017.

\end{thebibliography}
%
%
%
%


\end{document}